\documentclass[3p]{elsarticle}
\usepackage{graphicx}
\usepackage{enumerate}
\usepackage{latexsym,amssymb,amsthm,amsxtra}
\usepackage{amsmath}

\newtheorem{theorem}{Theorem}[section]

\newtheorem{prop}[theorem]{Proposition}

\newtheorem{assumption}[theorem]{Assumption}
\newtheorem{example}[theorem]{Example}

\def\S{{\cal S}}

\def\E{{\mathbb E}}

\def\V{{\cal V}}
\def\N{{\mathbb N}}

\def\A{{\mathcal{A}}}
\def\B{{\mathcal{B}}}

\def\G{{\mathcal{G}}}
\def\U{{\mathcal{U}}}

\def\sA{\mbox{\textsc{a}}}
\def\sB{\mbox{\textsc{b}}}
\def\sU{\mbox{\textsc{u}}}
\def\sS{\mbox{\textsc{s}}}
\def\sRX{\mbox{\textsc{rx}}}
\def\RX{\scriptsize{\mbox{\textsc{rx}}}}
\def\sTX{\mbox{\textsc{tx}}}
\def\TX{\scriptsize{\mbox{\textsc{tx}}}}
\def\V{{\mathcal V}}

\newcommand\cro[1]{\left\langle #1 \right\rangle}

\def\({{\Bigl(}}
\def\){{\Bigr)}}

\newcommand\esp[1]{{\mathbb E}\left\{#1\right\}}

\def\ind{{\mathchoice {\rm 1\mskip-4mu l} {\rm 1\mskip-4mu l}
{\rm 1\mskip-4.5mu l} {\rm 1\mskip-5mu l}}}
\def\square{\ifmmode\sqr\else{$\sqr$}\fi}
\def\sqr{\vcenter{
         \hrule height.1mm
         \hbox{\vrule width.1mm height2.2mm\kern2.18mm\vrule width.1mm}
         \hrule height.1mm}}                  

\title{Estimating the Spatial Reuse with Configuration Models}

\author{ P. Bermolen, M. Jonckheere, F. Larroca, P. Moyal}

\begin{document}

\begin{abstract}
We propose a new methodology to estimate the
spatial reuse of CSMA-like scheduling.
Instead of focusing on spatial 
configurations of users, we model the interferences
between users as a random graph. 
Using configuration models for random graphs, we show
how the properties of the medium access mechanism are captured by some deterministic differential equations, when the size of the graph gets large.
Performance indicators such as the probability of connection
of a given node can then be efficiently computed from these equations.
We also perform simulations to illustrate the results on different types of random graphs.
 Even on spatial structures, these estimates
get very accurate as soon as the variance of the interference 
is not negligible.
\end{abstract}

\maketitle

\section{Introduction}

Wireless communications are becoming ubiquitous. Nowadays, virtually every electronic device includes at least one wireless interface. With the massive expansion of Wireless Personal Area Networks (WPAN) and Wireless Sensor Networks (WSN), and the advent of the Internet of Things (IoT), this trend will only tend to augment. A recent study performed by Cisco in 2011~\cite{cisco2011} estimated in 1.84 the number of devices connected to the internet per person in the world (if only people that are actually connected to the internet are considered, this number rises to 6.25), and that it would be 6.58 by 2020. 
These modern networks,  decentralized and immense in size bring along several challenges 
regarding performance evaluation.

In the present article, we consider a large set of nodes who communicate with each other by means of a wireless channel. In this network, a Medium Access Control (MAC) mechanism based on 802.11's Distributed Coordination Function (DCF)~\cite{ieee80211} is in place to allow nodes to effectively share the medium. Since every node may be either receiver or transmitter indistinctly, the hidden node
 problem will probably degrade the network's performance if unattended~\cite{larroca201480211}. The typical way of dealing with this problem is to use the RTS/CTS four-way handshake. 

With this mechanism in place, every node that intends to send a packet to a tagged destination first senses the medium, and if idle, it sends a Ready To Send (RTS) frame addressed to it. This packet contains the duration of the impending transmission, so every node overhearing the packet will refrain themselves from transmitting during this period. If possible, the destination node will in turn respond with a Clear To Send (CTS) frame, which also contains the duration, and thus every other node overhearing this frame will also restrain themselves from transmitting. However, if the tagged destination has previously overheard another CTS frame, or its channel is currently not idle, it will naturally not answer the RTS, and the transmission will not take place. 


We are interested in the number of concurrent \emph{successful} transmissions that take place in such network, which is generally coined as the \emph{spatial reuse}. In particular, we will consider a slotted variant of DCF. That is to say, time is broken into slots of duration $T_s$, which in turn are separated in two periods: the so-called contention period ($T_c$) and the transmission period ($T_t$). During the first one, all nodes that have a packet ready to be sent choose a random time between $0$ and $T_c$, when they send an RTS frame to the destination node. Naturally, this will happen unless they sensed the medium as busy or they overheard another RTS or CTS frame first. The destination node will in turn answer with a CTS frame unless its medium is currently busy, or if it received another RTS or CTS frame first. After the RTS/CTS handshake successfully took place, the data packet is sent immediately. 

Before stating our contributions in more details, we mention some of the most relevant research directions in the large literature on performance evaluation of wireless networks driven by access control mechanisms.
The dynamics inferred by randomized decentralized medium access protocols have received a tremendous amount of attention, due on the one hand to their dominant deployment, and on the other hand to the great difficulty of their modeling and performance evaluation.
Even the simplest algorithms where users do not use any information on the contents of their own
or their neighbors buffers to access the channel are far from well understood.
A realistic model should indeed combine at least the two following interacting sources of randomness:
\begin{enumerate}
\item Randomness of the media: stochastic spatial positions and interference between users
competing for communication,
\item Randomness of traffic: stochastic arrivals and departures.
\end{enumerate}
Moreover, users can be considered saturated (they always have packets to transmit)
or unsaturated (i.e., exogeneous sources of traffic feed the network), while traffic can be asymmetric and subject to priorities mechanisms (tree-algorithms).
These difficulties gave rise to different types of models focusing on specific aspects and source of randomness:
\begin{itemize}
\item 
A first class of models does not take into account the spatial diversity and
considers scenarios where every user symmetrically interacts with everyone else.
 Some authors use the term full interference to describe this situation which might be a very pessimistic assumption. (In that context for instance, the non-adaptive ALOHA can be showed to be unstable for almost any traffic parameter \cite{bremaud}.) 
Some models may however focus on possible asymmetric and dynamical stochastic traffic characteristics.
This line of research started with the seminal work of Bianchi \cite{bianchi} 

\item  
Another family of models focuses on a fixed graph of interference with unsaturated users.
In this context, the first benchmark of performance is to characterize
the stability of the network which might turn out to be a difficult task. Under Markovian assumptions, a characterization of the stability problems has been obtained in \cite{Szpankowski83,stabST}
while it was shown in \cite{bordenave12} how to approximate the stability region using mean-fields arguments.

\item On the other side of the spectrum, some models inspired by stochastic geometry and point processes focuses on the random spatial location of users and aim at estimating e.g. the probability of connections for a given users' configuration. Since a limited spatial interaction is an important feature of wireless networks, it is crucial to grasp its quantitative influence on performance.
   To further take into account the stochastic nature of traffic, a time scale separation assumption might then be called upon to use  the probability of access in a given state of the system
 as the speed of service of a higher time scale stochastic network. 
\end{itemize}

\

This last point of view is the one we adopt here.
 The evaluation of connection probabilities for a given scenario remains however a difficult task. This is linked to the analysis of the dynamics of the so-called parking process, which has received a tremendous amount of attention in the physics and biology literature (under the name of random adsorption models) \cite{penrose,penrose_approximations}.
Even in the case of completely symmetric users, evaluating the probability of connection when users are spatially located as a Poisson point process in the plane is a difficult problem.
Recently, a characterization of the Laplace transforms of functionals of the parking process dynamics
(also called the Mattern-$\infty$ process \cite{BaccelliB09a,BaccelliB09b}) were obtained \cite{VietB12,BaccelliViet}.
These are striking results given the complexity of the process.
Unfortunately they do not provide explicit formula for the connection probabilities for instance and bounds have to be invoked to reach easily computable estimations.

\

We take here an alternative route.
Building on some recent results for parking processes on random graphs,
we base our estimation of connection probabilities on dynamics over random graphs rather than on thinned point processes.
Just as models inspired by stochastic geometry, we consider a single time slot but 
we forget the spatial configurations of users to focus on a random graph of interference. Though this method might loose in accuracy by ignoring some correlations between users interferences
it allows to get simple differential equations through which one can easily compute the probability of connection.

More specifically, our contributions are the following. Firstly, we provide a way to calculate the spatial reuse in a large arbitrarily random network (see Sec. \ref{sec:preliminaries} for a precise description of this randomness). This result is obtained by using and extending our previous work~\cite{aap} which, as discussed in Sec.\ \ref{sec:fluidlimit}, may be regarded as an analysis where the receiver node is not considered. As a first step, we provide an intuitive interpretation of the main result in~\cite{aap}. This later allows us to present several extensions that take into account the receiver and the RTS/CTS handshake. In particular, in Sec. \ref{sec:extensions} we will show how to analyze variants of the access mechanism, which differ in how the transmitter (and the overhearing nodes) react to a failed handshake. Moreover, the usefulness of the presented results is illustrated by studying several example scenarios (Sec. \ref{sec:examples}).


\section{Context and assumptions}\label{sec:preliminaries}


We assume in the sequel a threshold-type channel, where transmissions are either perfectly received or not received at all. This is sometimes termed \emph{protocol model} in the literature \cite{}. In other words, we consider an interference graph $\G=(\V,\mathcal{E})$, where an edge exists from node $s$ to node $r$ if the transmission of node $s$ is received by node $r$. We further assume that the channel is symmetric, meaning that all edges are bidirectional. Practically speaking, this means that all wireless nodes have relatively similar hardware and transmit at the same power. The last simplification is that the RTS/CTS handshake takes place instantly. Again, practically speaking, this means that thanks to the collision avoidance scheduling, no collisions will occur in our model. 

Let us now discuss how to model the communication graph. In a planned network, where all nodes are fixed and the propagation conditions are stable, the graph is given and does not change significantly over time. This could be the case in a wireless mesh with line-of-sight between nodes. However, if nodes are mobile and/or the channel rapidly changes (for instance, a MANET or a urban scenario), the graph may vary significantly, thus adding a second level of randomness (that is, in addition to the MAC layer). We do focus on the second scenario. 

Consider for instance the ``most random'' graph possible: a scenario where every pair of nodes are neighbors (i.e.\ an edge exists between them) with probability $p$, and the event of two nodes being neighbors is independent of everything else. This graph models a totally unplanned and dynamic network, where all we may be able to estimate a priori is $p$. In particular, let us note as $N$ the total number of nodes in the network, and as $\nu$ the average number of neighbors of each node. Then, $p(N-1)=\nu$, and if $N$ is big, this amounts to $p = p_N = \nu/N$. Thus, the parameter of interest in this scenario is $\nu$. 

We may however have more information on the communication graph than the mean number of neighbors. In particular, we will assume here that the graph of interference is characterized
by the complete distribution of degrees denoted $h_N$.
For example, we may know that half the nodes have two neighbors, and the other half three, in which case the counting measures of degrees is such that $h_N(2)=h_N(3)=1/2$, and 0 for the rest. 
For the example described in the previous paragraph the degree distribution is $h_N(i)=C_i^{N-1}p^i(1-p)^{N-1-i}$ (i.e. a binomial distribution), which may be approximated by a Poisson distribution with parameter $\nu$ when $N$ is large. 

Regarding traffic, we assume that all nodes are saturated i.e. have a packet ready to be sent in every time-slot. We further assume that the destination node for this packet is a neighbor picked at random. Let us consider a given contention period. At time 0, every node will choose a random number, uniformly distributed between 0 and $T_c$. Consider the node with the minimum such time. It will send a RTS frame to one of its neighbors, chosen randomly among them. Since this is the first transmission, the destination node will answer with a CTS frame, thus ``blocking'' all its neighbors. The origin node will immediately start transmitting a data frame, also blocking all its neighbors. In what follows, we will term these two nodes as \emph{active}, and their neighbors as \emph{blocked}. Just like the receiving node, blocked nodes stop competing for the channel. Let us term the rest of the nodes as \emph{unexplored}. 

We now have to focus solely on this last set, and find the node who has drawn the minimum number. This node, say $i$, will then send a RTS frame to any of its neighbors, say $j$. However, the chosen node may already be blocked by a previous transmission and/or CTS, in which case the handshake will fail.  As a first step in the analysis, and in the following two sections, we will assume that the origin node realizes this failure (since no CTS is received), and immediately sends a new RTS frame to another random neighbor. This is repeated until a CTS is received back, or no more neighbors are left (i.e.\ all its neighbors belong to the blocked set). Furthermore, these potentially several transmissions of RTS frames will also be assumed to take place instantly. In addition to allowing us to illustrate the techniques we will use in the analysis, this scenario has a value of its own, since it models an ideal protocol and may be regarded as an upper bound to the resulting spatial r
 euse. 

Finally, the above procedure is repeated until time $T_c$, or equivalently, until no more unexplored nodes are left. 

\section{Preliminaries}
\subsection{Random sequential adsorption}\label{subsec:naive}

We now start the analysis with the objective of estimating the number of successful transmissions that take place concurrently. That is to say, how many CTS frames are sent in a single time slot in average. The first step will be to construct the interference graph $\G=(\V,\mathcal E)$. Since our only a priori information is the nodes' degree distribution $h_N(i)$, we will uniformly choose a graph among those that comply with it. 

Having chosen a graph, we proceed to analyze the MAC protocol. With the description of the access mechanism we presented in the previous section in mind, let us note as $\{T_i\}_{i=1,\ldots,N}$ the times chosen by the nodes to send its RTS frame. These are drawn uniformly from the interval $[0,T_c]$. However, regarding the order at which they will proceed, which ultimately dictates the successful transmissions that will take place, any other continuous distribution is equivalent. In particular, the exponential distribution with parameter $\lambda$ is one that will serve our purposes. Let us now note as $\U_t$, $\A_t$ and $\B_t$ the set of unexplored, active and blocked nodes at time $t$ respectively (with $\U_t \cup \A_t \cup \B_t = \{1,\ldots,N\}\,\forall t$). At time 0 we have $\A_0=\B_0=\emptyset$ and $\U_0=\{1,\ldots,N\}$. 

The first transmission attempt will happen at the minimum of all the $N$ competing exponentials. Let $t_1$ be such time. The distribution of this random variable will thus be exponential with parameter $N\lambda$, and the transmitting node will be uniformly chosen from $\U_t$. Let $s$ be this node. The first step then is to update the sets as follows: $\A_{t_1^+}\leftarrow \A_{t_1}\cup \{s\}$ and $\U_{t_1^+}\leftarrow \U_{t_1}\setminus \{s\}$. 

We now have to choose a random neighbor of $s$, and check whether it will be able to answer with a CTS. In our present ideal scenario, this will be repeated until either a CTS is received, or no more neighbors are left. This amounts simply to checking whether $s$ currently has any unexplored neighbors, and if so, choosing one randomly. Therefore, from the neighbors of $s$ (the set $\mathcal{N}_{s}$ obtained from $\G$), we choose a random node from $\mathcal{N}_{s}\cap\U_{t_1^+}$ (if any), which we will denote as $r$. The following two cases are then possible:
\begin{enumerate}
    \item If no such $r$ exists (i.e.\ $\mathcal{N}_s\cap\U_{t_1^+}=\emptyset$) we proceed to the next node. 
    \item If a $r$ in $\U_{t_1^+}$ was found, we re-update the set of active nodes as $\A_{t_1^{+}}\leftarrow \A_{t_1^+}\cup \{r\}$. The set of neighbours of $r$ will be noted as $\mathcal{N}_{r}$, and it will naturally include $s$. We thus finally update the blocked set of nodes as $\B_{t_1^+}\leftarrow \B_{t_1}\cup (\mathcal{N}_{s}\setminus \{r\})\cup (\mathcal{N}_{r}\setminus \{s\})$, the unexplored one as $\U_{t_1^+}\leftarrow \U_{t_1^+}\setminus \left(\mathcal{N}_{s}\cup\mathcal{N}_{r}\right)$, and proceed to the next node. Figure \ref{fig:ex_vecinos} illustrates this case for a generic time $t_n$.
\end{enumerate}

 \begin{figure}
     \centering
     \includegraphics[scale=0.5]{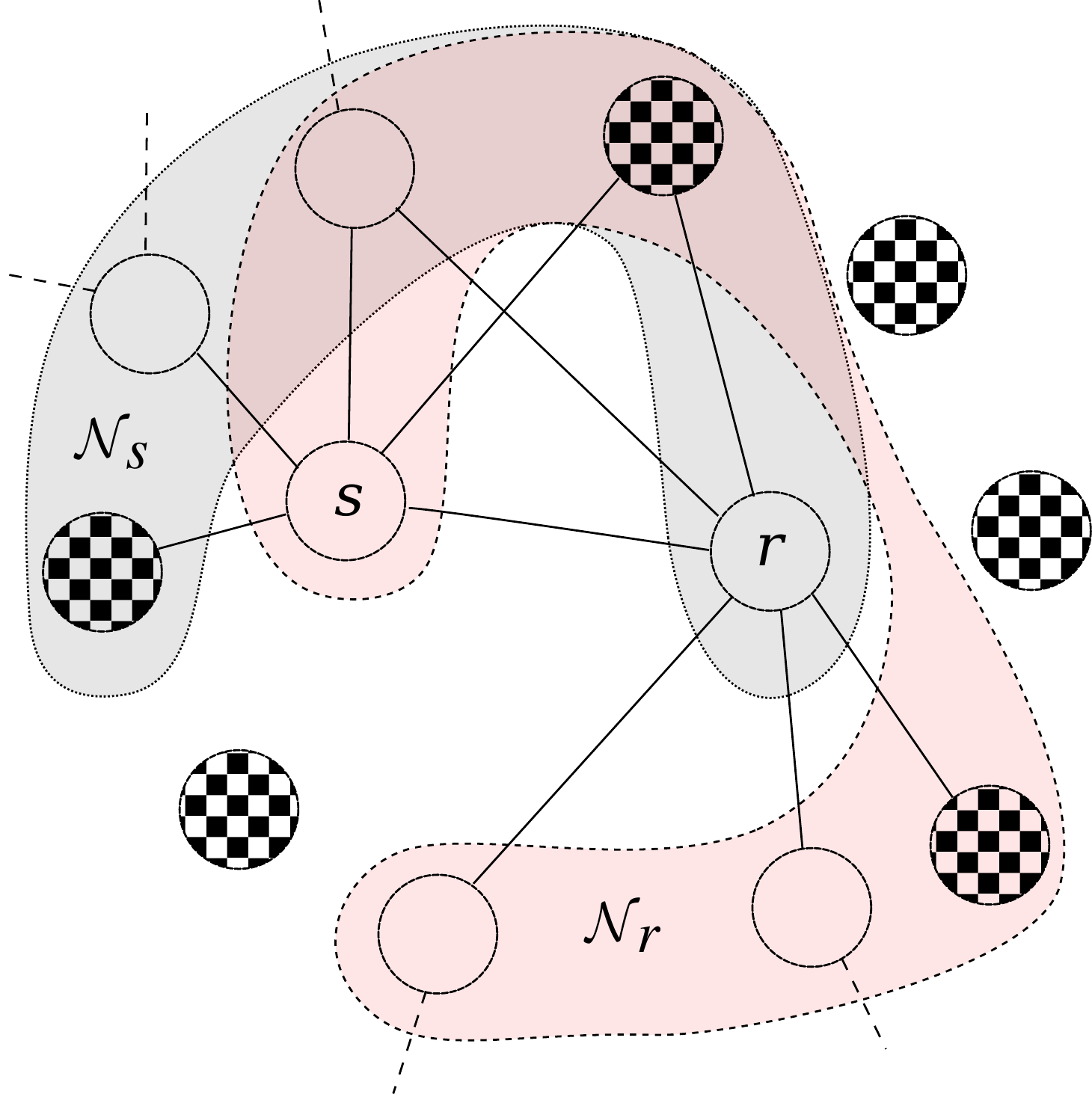}
     \caption{An example of $\mathcal{N}_{s}$ and $\mathcal{N}_{r}$ when a $r$ in $\U_{t_n^+}$ exists. Nodes with no fill were unexplored at time $t_n$. Only the connections of $s$ and $r$ are shown. }\label{fig:ex_vecinos}
 \end{figure}

The next transmission attempt (if any) will happen at time $t_2$. Since only unexplored nodes compete for the channel, the random variable $t_2-t_1$ follows an exponential distribution with parameter $|\U_{t_1^+}|\lambda=|\U_{t_2}|\lambda$. The procedure described above is still valid for $t_2$ (and for every other $t_n$). It is important to highlight that, since the tagged node is in the set of unexplored nodes, it is not a neighbor of any active node, and thus all its neighbors will necessarily belong to the set $\U_{t_2^+}\cup\B_{t_2}$ (i.e.\ excluding $\A_{t_2^+}$). The above process will be repeated until no more unexplored nodes are left: let $t_{n^*}$ be this time (i.e.\ $\U_{t_{n^*}}=\emptyset$). 

The algorithm described above belongs to the class of the so-called \emph{Parking Processes}~\cite{penrose}. In this context, the set $\A_{t_{n^*}}$ is termed \emph{Jamming limit}.

Please note that, given $\G$, the triplet $(\U_t,\A_t,\B_t)$ forms a non-homogeneous continuous time Markov chain. Moreover, any state in which $\U_t=\emptyset$ is an absorbing one, and the total number of transitions that it takes to reach any such state is the number of transmissions attempts at the given contention period (i.e.\ the number of RTSs sent). Moreover, the cardinality of $\A_t$ when the chain reaches this state is the number of RTSs and CTSs sent. We have thus proved the following relation: 
\begin{gather*}
    \E\{\Theta\}=\theta=\underset{t\rightarrow \infty}{\lim}\E\left\{\frac{\left|\A_t\right|}{N}\right\}-\E\left\{\frac{n^*}{N}\right\}, \label{eq:exp_Theta}
\end{gather*}
where $\Theta$ is a random variable indicating {\bf the number of successful transmissions} obtained in a given contention period, divided by the total amount of nodes ($N$), and $n^*$ is a random variable indicating {\bf the number of transitions} it took the chain to reach an absorbing state. 

One may try to directly define a Markov chain on a fixed graph and
 calculate the two terms that appear in \eqref{eq:exp_Theta} given a particular communication graph $\G$, and then combine these values weighting them by the probability of every particular graph. Unfortunately, this brute-force approach will quickly be limited by the immense 
size of the state space of the Markov chain.
This explains the need for an alternative, more tractable, approach.



\subsection{Configuration algorithm and measure-valued Markov process for large graphs
}\label{subsec:conf_algo}

Given the difficulties emphasized above, our approach consists in constructing
 the random graph $\G$ jointly with the transmissions' dynamics, following the so-called configuration model \cite{bollobas01,molloyreed,durrett}. In a first instance, this construction focuses only in the unexplored nodes. In fact, we are interested in how many CTSs were sent when the unexplored set $\U_t$ becomes empty. 

To this end, let us recall that at any time $t$, a node can be an element of either $\U_t$, $\A_t$ or $\B_t$. Moreover, our only a priori information is the  counting measure $\mu$, where $\mu(i)$ is the nu;ber of unexplored nodes having $i$ unexplored neighbors at time $0$. We then extend this to an arbitrary time $t>0$ and denote it as $\mu_t$. 

It is important to highlight that we do {\bf not} require the knowledge of which particular nodes are the neighbors of any given unexplored node. As such, we regard the edges starting from any node in $\U_t$ as \emph{unmatched} half-edges from an unexplored node towards either another unexplored node (and we denote them as $\sU\to\sU$) or a blocked node (which will naturally be noted as $\sU\to\sB$). In this sense, $\mu_t(i)$ tells us how many unexplored nodes have $i$ half-edges of class $\sU\to\sU$ at time $t$. As we now discuss, by characterizing $\mu_t$ alone we will be able to estimate the spatial reuse. 

Let us denote for all $t \ge 0$ and all $i \in \N$, 
\begin{align}
\label{eq:defalpha}
\alpha_t(i)&=\frac{\mu_t(i)}{\sum_{j\in \N}\mu_t(j)};\\
\label{eq:defbeta}
    \beta_t(i)&=\frac{i\mu_t(i)}{\sum_{j \in \N} j\mu_t(j)}.
\end{align} 
The probability measures $\alpha_t$ et $\beta_t$ have 
intuitive interpretations: the first one describes the degree distributions of a randomly (and uniformly) picked unexplored node at time $t$,
while the second one is the size biased distribution of $\alpha_t$
and represents the degree distribution of any neighbors of a randomly picked unexplored node or in other words, the degree of the starting node 
of an half-edge drawn uniformly at random, among all half-edges starting from unexplored nodes. This representation will be useful 
in the following discussion.   

Suppose that at a transition time $t_n$, a formerly unexplored node transmits a RTS frame, thus becoming active. Let us term this node as $\sTX$. Then, the measure $\mu_{t_n}$ has to be modified according to the following steps: 
\begin{enumerate}
    \item $\sTX$ is chosen uniformly among all unexplored nodes. We shall denote by $K^{\TX}$ {\bf its degree towards other unexplored nodes}. Hence, 
        \begin{gather}
            P(K^{\TX}=i)=\alpha_{t_n}(i).\label{eq:def_alpha}
        \end{gather}
         (Note that we have dropped the dependence on time for $K^{i}$ to lighten notations
         and we shall do the same for other quantities of interest in the sequel.)
         
        Then, as $\sTX$ becomes active and is no longer unexplored, we have to remove it from the measure $\mu_{t_n}$: the quantity $\mu_{t_n}(i)$ is decreased by one in $i=k^{\TX}$, a particular realization of $K^{\TX}$ (for example $k^{\TX}=3$ in Fig.\ \ref{fig:ex_vecinos}). If $k^{\TX}=0$, then this transmission attempt will fail and we proceed to the next node, if any. 

    \item If $k^{\TX}>0$, one of the unexplored neighbors of $\sTX$ also becomes active: 
let $\sRX$ be the intended destination of $\sTX$'s RTS. One of the $k^{\TX}$ half-edges starting from $\sTX$ is matched with an $\sU\to\sU$ half-edge, 
drawn uniformly at random among the $\sum i\mu_{t_n}(i) - k^{\TX}$ available ones, in order to complete a full edge between $\sTX$ and $\sRX$. 

If $k^{\TX}>1$, we also need to block all the unexplored neighbors of $\sTX$ except $\sRX$, and for this we repeat the same random matching procedure 
$k^{\TX}-1$ times. 
  
In total, $k^{\TX}$ $\sU$-nodes change status at that stage, and we have to know their respective degrees to update the measure accordingly. 
Observe that at $t_n$, precisely $i\mu_{t_n}(i)$ half-edges belong to a node that has a degree $i$ towards unexplored nodes. Let for any $i$, $Y^{\TX}(i)$ the number of neighbors of $\sTX$ (including $\sRX$) having degree $i$ toward the unexplored vertices just before $t_n$. In view of the previous observation, conditionally to $K^{\TX}$, $Y^{\TX}(i)$ is an hypergeometric random variable with parameters $K^{\TX}$ and $\beta_{t_n}(i)$ among all the available half-edges at $t_n$. All in all, for any $i$, $\mu_{t_n}(i)$ decreases by the quantity $y^{\TX}(i)$, a realization of $Y^{\TX}(i)$. 

\item Let us denote by $K^{\RX}$, the  {\bf number of unmatched} $\sU\to\sU$ {\bf half-edges starting from} $\sRX$ (for example, $k^{\RX}=4$ in Fig.\ \ref{fig:ex_vecinos}). Set $K^{\RX} \equiv 0$ if $k^{\TX}=0$.  
In view of the previous discussion, conditionally to $\{K^{\TX}>0\}$ the distribution of the r.v. $K^{\RX}$ is given by
        \begin{gather}
            P(K^{\RX}=i) = \beta_{t_n}(i),\,i\in\N^*.\label{eq:def_beta}
        \end{gather}
        
Likewise the previous step, if $k^{\RX}>1$, we now have to block all the unexplored neighbors of $\sRX$, except $\sTX$ and all its (already blocked) neighbors. Let $Y^{\RX}(i)$ be the number of such neighbors of $\sRX$, having degree $i$ toward the unexplored vertices at $t_n$. 
The distribution of the latter r.v. depends on the number of neighbors that $\sTX$ and $\sRX$ have in common. Such neighbors may very well exist, 
as the example in Fig.\ \ref{fig:ex_vecinos} illustrates. But recall that we are interested in a large network, where $N$ goes to infinity. In such case, the probability of $\sRX$ choosing one (or more) of the (finitely many) neighbors of $\sTX$ is arbitrarily small. Therefore, to the limit, $Y^{\RX}(i)$ is an hypergeometric random variable with parameters $k^{\RX}-1$ and $\beta_{t_n}(i)$. As above, for any $i$, $\mu_{t_n}(i)$ additionally decreases by the quantity $y^{\RX}(i)$. 

    \item 
    Once we have blocked all the neighbors of both $\sTX$ and $\sRX$, all their formerly $\sU\to\sU$ half-edges disappear and become either $\sB\to\sA$ (the ones pointing towards either $\sTX$ or $\sRX$) or $\sB\to\sU$ (the rest, see the slashed edges in Fig.\ \ref{fig:ex_vecinos}). This has been performed in the previous step. However, and regarding the latter half-edges, there will be a matching number of $\sU\to\sU$ half-edges from other nodes that now become $\sU\to\sB$, and thus have to be removed from $\mu_{t_n}(i)$. This means that, following the example in Fig.\ \ref{fig:ex_vecinos}, five unexplored nodes (each with a random degree towards unexplored nodes with distribution $\beta_{t_n}(i)$) will ``lose'' one $\sU\to\sU$ half-edge, and the measure $\mu_{t_n}$ should be updated accordingly. We then proceed to the next node, if any. 
   
       Notice that the previous discussion still applies here: each blocked node completes its formerly $\sU\to\sU$ half-edges with any of the available half-edges, uniformly chosen. Since there are infinitely many, {\bf the probability of choosing precisely one belonging to one of the rest of the blocked neighbors is arbitrarily small when the number of nodes is large}. This means that the neighbors of $\sTX$ or $\sRX$ will not be neighbors between them. Furthermore, and for the same reason, each of the five slashed half-edges in Fig.\ \ref{fig:ex_vecinos} will point towards a different node.\footnote{A formal proof of this intuitive analysis may be found in \cite{aap}.}
       
       Hence, in terms of the random vectors defined before, there are 
$Z^{\TX}=\sum_{\ell>0}(\ell-1)\left(Y^{\TX}(\ell)-\ind_{K^{\RX}=\ell}\right)^+$ (respectively, $Z^{\RX}=\sum_{\ell>0}(\ell-1)Y^{\RX}(j)$ whenever $k^{\TX}>0$) unexplored nodes neighboring the blocked neighbors of $\sTX$ (resp., of $\sRX$) that lose one $\sU\to\sU$ half-edge (in the example of Fig.\ \ref{fig:ex_vecinos} $z^{\TX}+z^{\RX}=5$). Let us define for all $i$, $X^{\TX}(i)$ (resp., $X^{\RX}(i)$) the r.v. indicating {\bf how many such nodes have a degree $i$ towards the unexplored nodes}. The latter is hypergeometric with parameters $Z^{\TX}$ (resp., $Z^{\RX}$) and $\beta_{t_n}(i)$. Using these definitions, in this last step, for any $i$, $\mu_{t_n}(i)$ decreases by $x^{\TX}(i)+x^{\RX}(i)$ and increases by $x^{\TX}(i+1)+x^{\RX}(i+1)$. 
\end{enumerate}

%

The above description shows that $\left(\mu_t\right)$ is a measure-valued continuous-time inhomogeneous Markov chain (\cite{dawson}), admitting 
the null measure $\mathbf 0$ as absorbing state.

Let us further add another dimension to the process which counts how many CTS frames have been sent up to time $t$, which we  denote as $c_t$. 
It is easily seen that $(\mu_t,c_t)$ is still a Markov chain. We are thus interested in $c_t$ after the chain reaches an absorbing state, considering that the following equality holds: 
\begin{gather*}
    \theta=\underset{t\rightarrow \infty}{\lim}\E\left\{\frac{c_t}{N}\right\}
\end{gather*}

\section{Large graph limit}\label{sec:fluidlimit}


As we discussed in the previous section, the spatial reuse can be described as a function of the Markov process $(\mu_t,c_t)$.  
Though computations on the defined Markov chain are still a formidable task for large state spaces,
 the system dynamics get simpler to understand, for large $N$. 
In fact,  $c_t$ can be expressed in the limit by means of a differential equation (where naturally $\mu_t$ plays a central role).

Let us emphasize the dependence of the measure-valued process on the size $N$ of the graph by denoting the latter 
$\left(\mu^N_t\right)$. As $N$ goes large, we consider a scaled process $\left(\bar\mu_t^N\right):=\left(\mu^N_t/N\right)$ (thus for all $i$, $\bar\mu_t^N(i)$ refers to the proportion of nodes with $i$ unexplored neighbors at time $t$ for the graph of size $N$). We then identify the so-called large-graph limit $\bar\mu$ by letting $N$ go to infinity, thereby obtaining a functional law of large numbers for the system. This limiting process gives us a representation of the ``mean behaviour'' of the measure valued process $(\mu_t)$, and the spatial reuse as a by-product. 

In what follows, for any $t$, for any fixed measure $\bar\mu_t$ we denote likewise (\ref{eq:defalpha}) and 
(\ref{eq:defbeta}),  
\begin{align}
\label{eq:defbaralpha}
\bar\alpha_t(i)&=\frac{\bar\mu_t(i)}{\sum_{j\in \N}\bar\mu_t(j)}=\frac{\bar\mu_t(i)}{\cro{\bar\mu_t,1}};\\
\label{eq:defbarbeta}
 \bar\beta_t(i)&=\frac{i\bar\mu_t(i)}{\sum_{j \in \N} j\bar\mu_t(j)}=\frac{i\bar\mu_t(i)}{\cro{\bar\mu_t,\chi}}.
\end{align}

\subsection{Comeback to the parking process}
 
Let us re-interpret the convergence results of \cite{aap}, which shall turn out to be useful to derive the corresponding results in the present context. 
Let for any $k \ge 0$, $\chi^k$ denotes the function $x \to x^k$. The suitable initial assumptions under which the convergence results of \cite{aap} (and thereby, the ones presented hereafter) hold true, are given below: 

  

\begin{assumption}
\label{hypo:1}
For all bounded functions $\phi$ from $\N$ to itself,
\begin{equation}
\label{eq:convinit}
\cro{\mu^n,\phi} \underset{n\rightarrow \infty} {\overset{\mathcal{(P)}}{\to}}\cro{\nu,\phi},
\end{equation}
where $\nu$ is a deterministic finite measure on $\mathbb N$ such that for some constants $\alpha>0$ and $M>1$,
\begin{equation}
\label{eq:condinit}
\alpha<\cro{\nu,\chi}\,\mbox{ and }\,\cro{\nu,\chi^6} <M.
\end{equation}
\end{assumption}

\

It is proven in Theorem 3.1 of \cite{aap} that under Assumptions \ref{hypo:1}, the sequence of processes $\left\{\bar\mu^N\right\}$ converges in probability 
and uniformly over compact time intervals to the unique deterministic function $\bar\mu$ satisfying for  
 all $i\in\N$, 
\begin{gather}
    \frac{d}{dt}\bar\mu_t(i) = 
    \begin{cases}
        -\lambda \left[\bar\mu_t(i)+i\bar\mu_t(i) + (i\bar\mu_t(i)-(i+1)\bar\mu_t(i+1))\left(\frac{\sum_j j^2\bar\mu_t(j)}{\sum_j j\bar\mu_t(j)}-1\right)\right]& \text{ if }\sum_j j\bar\mu_t(j)>0;\\
                                                                                                                                -\lambda\bar\mu_t(0)1_{\{i=0\}} & \text{ if }\sum_j j\bar\mu_t(j)=0. 
    \end{cases}
\label{eq:limflu_i}
\end{gather}
As a matter of fact, this set of equations can be interpreted quite easily,  
by relating it to the discussion in Sec.\ \ref{subsec:conf_algo}. In particular, since an access mechanism with no destination node is considered, we can recognize the steps (1), (3) and (4) in Sec.\ \ref{subsec:conf_algo} (with all references to $\sRX$ omitted). Recalling the definitions 
therein, we obtain that for all $i$, 
\begin{align}
    \frac{d}{dt}\bar\mu_t(i) & = 
-\lambda \sum_j\bar\mu_t(j)\left[\bar\alpha_t(i) + \mathbb{E}\{K^{\TX}\}\bar\beta_t(i) + \mathbb{E}\{K^{\TX}\}(\bar\beta_t(i)-\bar\beta_t(i+1))\sum_j(j-1)\bar\beta_t(j)\right].\nonumber\\
   & =  -\lambda \sum_j\bar\mu_t(j)\left[\bar\alpha_t(i) + \mathbb{E}\{Y^{\TX}(i)\} + \mathbb{E}\{X^{\TX}(i)\} - \mathbb{E}\{X^{\TX}(i+1)\}\right].
   \label{eq:transformlimflui}
\end{align}

\subsection{Sender-Receiver algorithm}

We may now turn back our attention to the more realistic scenario where the sender chooses a receiver among its
neighbors which in turn might disable its own neighbors to be able to receive the message.
In this context, we have the following result:


\begin{theorem}
Under Assumptions \ref{hypo:1}, the sequence of processes $\{\bar\mu^N\}$ converges in probability and uniformly on compact time intervals 
towards the only measure-valued deterministic function $\bar \mu$ of the following infinite dimensional differential system: 
for all $t\ge 0$ and all $i\in\N$,
\begin{align}
\bar\mu_0(i)&=\nu(i)\nonumber;\\
{d \over dt} \bar\mu_t(i) &= -\lambda \sum_j\bar\mu_t(j) \Bigg[ \bar\alpha_t(i) + \bar\beta_t(i) 
\Big(\sum_jj\bar\alpha_t(j)+(1-\bar\alpha_t(0))\sum_j(j-1)\bar\beta_t(j)\Big)
\Bigg. \nonumber\\
                          &+ \left.\Big((\bar\beta_t(i)-\bar\beta_t(i+1))\sum_j (j-1)\bar\beta_t(j)\Big)\Big(\sum_jj\bar\alpha_t(j)+(1-\bar\alpha_t(0))\sum_j(j-2)\bar\beta_t(j)\Big)\right]. 
\label{eq:ecdif_ideal}
\end{align}
\end{theorem}

\begin{proof}
The steps to 
rigorously prove the uniqueness of the solution $\bar\mu$, or the convergence to that solution,
are exactly similar to the ones followed in \cite{aap} and we do not reproduce them here.
We can then extrapolate from the system (\ref{eq:transformlimflui}), that the right-hand of the equation above should be the mean number of nodes with $i$ half-edges of type $\sU\to\sU$ that are removed at time $t$, times the normalized total transition rate. That is to say, for all $i$, 
\begin{multline}
    {d\over dt}\bar\mu_t(i) = -\big(\lambda\sum_j\bar\mu_t(j)\big)\big[\bar\alpha_t(i)+\esp{Y^{\TX}(i)}+\esp{X^{\TX}(i)-X^{\TX}(i+1)} \big.\label{eq:ecdif_ideal_exp}
    \\ + \big.P(K^{\TX}>0)\left(\esp{Y^{\RX}(i)} + \esp{X^{\RX}(i)-X^{\RX}(i+1)}\right)\big].
\end{multline}

%

The first terms between brackets correspond to the transmitter node, exactly as in (\ref{eq:transformlimflui}). If $K^{\TX}>0$, then 
additional nodes of degree $i$ may be blocked by $\sRX$, which correspond to the rest of the terms. 
According to the discussion of Section \ref{subsec:conf_algo}, the mean expectations of the involved random variables are given by 
\begin{align*}
    P(K^{\TX}>0) &= 1-\bar\alpha_t(0);\\
    \E\{Y^{\TX}(i)\} &=\esp{K^{\TX}} \bar\beta_t(i);\\
    \E\{Y^{\RX}(i)\} &=\esp{\left(K^{\RX}-1\right)^+}\bar\beta_t(i);\\
    \mbox{for all }j,\,\E\{X^{\TX}(j)\} &= \sum_{\ell>0}(\ell-1)\esp{\left(Y^{\TX}(\ell)-\ind_{K^{\RX}=\ell}\right)^+}\bar\beta_t(j);\\
    \mbox{for all }j,\,\E\{X^{\RX}(j)\} &= \sum_{\ell>0}(\ell-1)\esp{Y^{\RX}(\ell)}\bar\beta_t(j),
\end{align*}
and \eqref{eq:ecdif_ideal} follows by straightforward algebra. 
\end{proof}

\begin{example}

Before discussing how to calculate the spatial reuse, let us consider an example. In particular, assume a large network where all we know is that nodes may have either 1, 2 or 3 neighbors with the same probability. In this case, $\mu_0^N(i)=1/3$ for $i=1,2,3$, and $\mu_0^N(i)=0$ for the rest. The same will be true for $\bar\mu_0$. It is then straightforward to see that in this case the infinite dimensional differential equation described by Eq.\ \eqref{eq:ecdif_ideal} results in a 4-dimensional one, where we have to determine $\bar\mu_t(i)$ for $i=0,\ldots,3$, and the rest will be identically zero. In general, if $\bar\mu_0(i)=0$ for all $i>D$, then the system has $D+1$ differential equations with $D+1$ functions to be calculated. In this case, we may resort to numerical methods to solve the system. 

The resulting solution $\bar\mu_t$ and several realizations of $\mu_t^N$ (with $N=1000$) for this particular example are compared in Fig.\ \ref{fig:ex_ideal_mu}. The graph illustrates how $\bar\mu_t$ effectively represents the ``mean'' of the actual process $\mu^N_t$. 

\end{example}

 \begin{figure}
     \centering
     \includegraphics[scale=0.4]{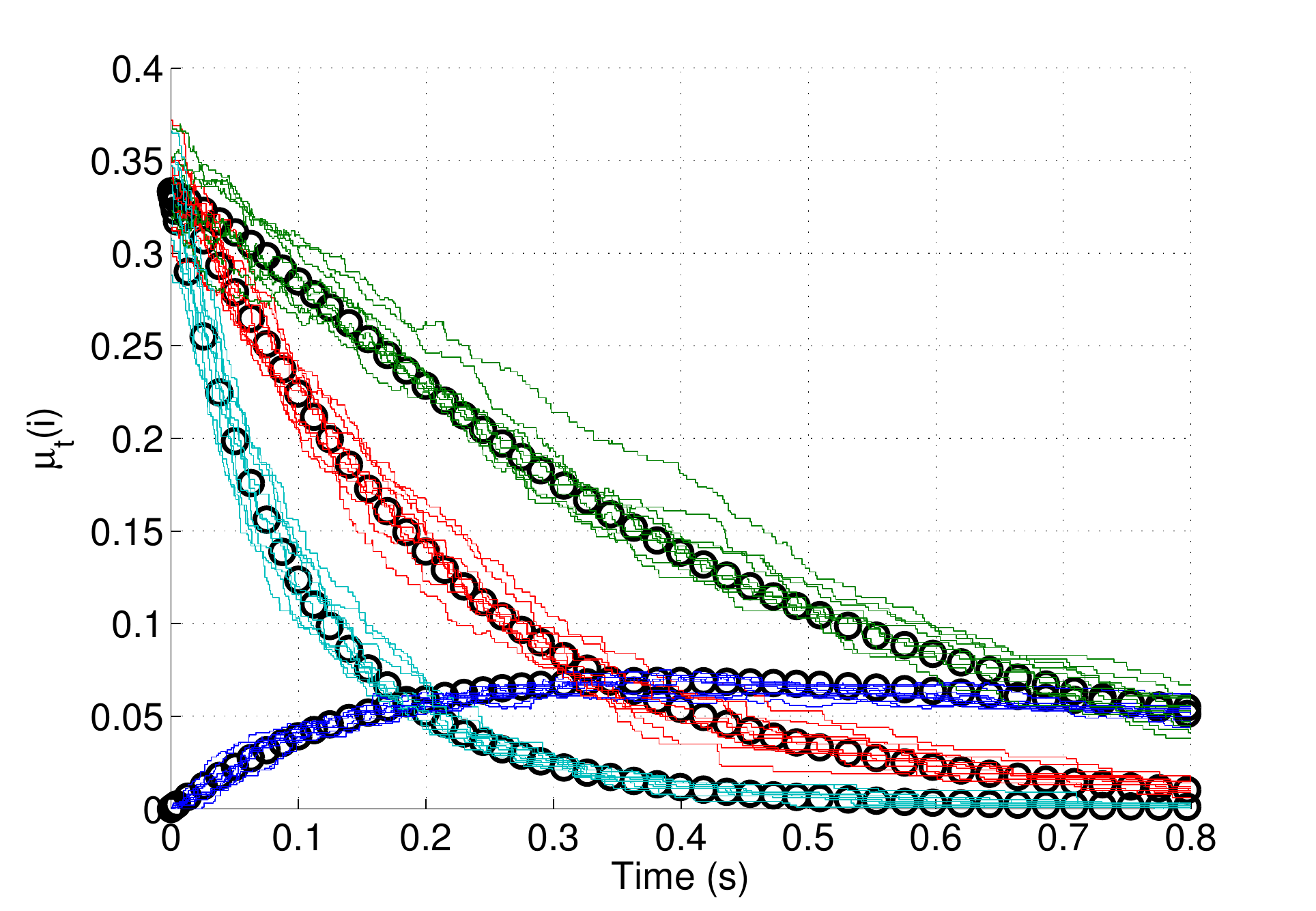}
     \caption{An example comparing several realizations of $\mu_t^N$ ($N=1000$) and the solution of \eqref{eq:ecdif_ideal} (marked as circles), where $\mu_0^N=\bar\mu_0=(0,1/3,1/3,1/3)$. The four coordinates of the process are shown overimposed.}\label{fig:ex_ideal_mu}
 \end{figure}

\subsection{Estimating the spatial reuse}\label{subsec:spatial_reuse}

Let us now discuss how to calculate the spatial reuse based on $\bar\mu_t$. As discussed before, we may define an auxiliary Markov chain $c_t$ that counts how many CTS frames have been sent up to time $t$. Its transition times are the same as $\mu_t$, and it will increase by one only if a CTS is received by the tagged node (i.e.\ if it has an unexplored neighbor) which occurs with probability $1-\alpha_{t}(0)=1-\mu_{t}(0)/\sum_j\mu_{t}(j)$. Based on the result presented before, 
 $c_t/N$ converges when $N$ goes to infinity to $\bar c_t$ which is deterministic and given by the following differential equation: 
\begin{gather*}
    {d \over dt}\bar c_t = \lambda\sum_j\bar\mu_t(j) \left(1-\frac{\bar\mu_t(0)}{\sum_j \bar\mu_t(j)}\right) = \lambda \sum_{j>0}\bar\mu_t(j), 
\end{gather*}
where $\bar\mu_t$ is obtained from the solution of \eqref{eq:ecdif_ideal}. Since the spatial reuse is simply the limit in $t$ of $\bar c_t$, we have obtained the following result:
\begin{prop}
    Let $\Theta$ be a random variable indicating the number of successful transmissions that take place in a given contention period, as described in Secs.\ \ref{subsec:naive} or \ref{subsec:conf_algo}. When $N$ goes to infinity the following equality holds: 
    \begin{gather}\label{eq:sr_ideal}
        \E\{\Theta\} = \theta = \lambda\int_0^\infty \sum_{j>0}\bar\mu_t(j)dt,
    \end{gather}
    where $\bar\mu_t$ be the solution to \eqref{eq:ecdif_ideal}. Equivalently, let $\bar u_t=\sum_i\bar\mu_t(i)$ be the proportion of unexplored nodes at time $t$ and $P_t(CTS)$ be the probability of receiving a CTS at time $t$ (in this case $P_t(CTS)=1-\bar\alpha_t(0)$). We may re-write Eq.\ \eqref{eq:sr_ideal} as: 
    \begin{gather}\label{eq:sr_pcts}
        \theta = \lambda\int_0^\infty \bar u_tP_t(CTS)dt.
    \end{gather}
\end{prop}

\section{Spatial reuse estimation over different interference graphs}\label{sec:examples}

In this Section, we first look at the precision of the approximation of the spatial reuse on configuration
models with a finite number of nodes. Later on, we show how this methodology may be efficiently used to estimate the spatial reuse of more complicated interference graphs steaming from spatial models.

\subsection{Configuration model with a uniform distribution}

As a first example to illustrate the precision of Eq.\ \eqref{eq:sr_ideal} for finite large $N$, let us consider the example in Fig.\ \ref{fig:ex_ideal_mu}. More precisely, we suppose that the network has a number of neighbors ranging from $5-k$ to $5+k$ (with $k$ varying between 0 and 5), all with the same probability. Figure \ref{fig:ex_ideal_sr} compares the limiting value $\theta$ (obtained from the solution of the limiting differential equation) and the simulations of 10 contention periods (in the form of a boxplot) for $N=1000$. 

Eq.\ \eqref{eq:sr_ideal} provides an excellent approximation to the mean of $\Theta$. Please note that the mean number of neighbors of each node is always 5, independently of $k$, which may be regarded as a parameter that controls the variance of the initial degree distribution. Interestingly, note that for low values of $k$ the spatial reuse seems independent of $k$ whereas when it reaches $5$, the spatial reuse significantly decreases. 
{\bf As intuition tells, this particular example illustrates that the spatial reuse does not depend only on the mean number the degree distribution.}

 \begin{figure}
     \centering
     \includegraphics[scale=0.4]{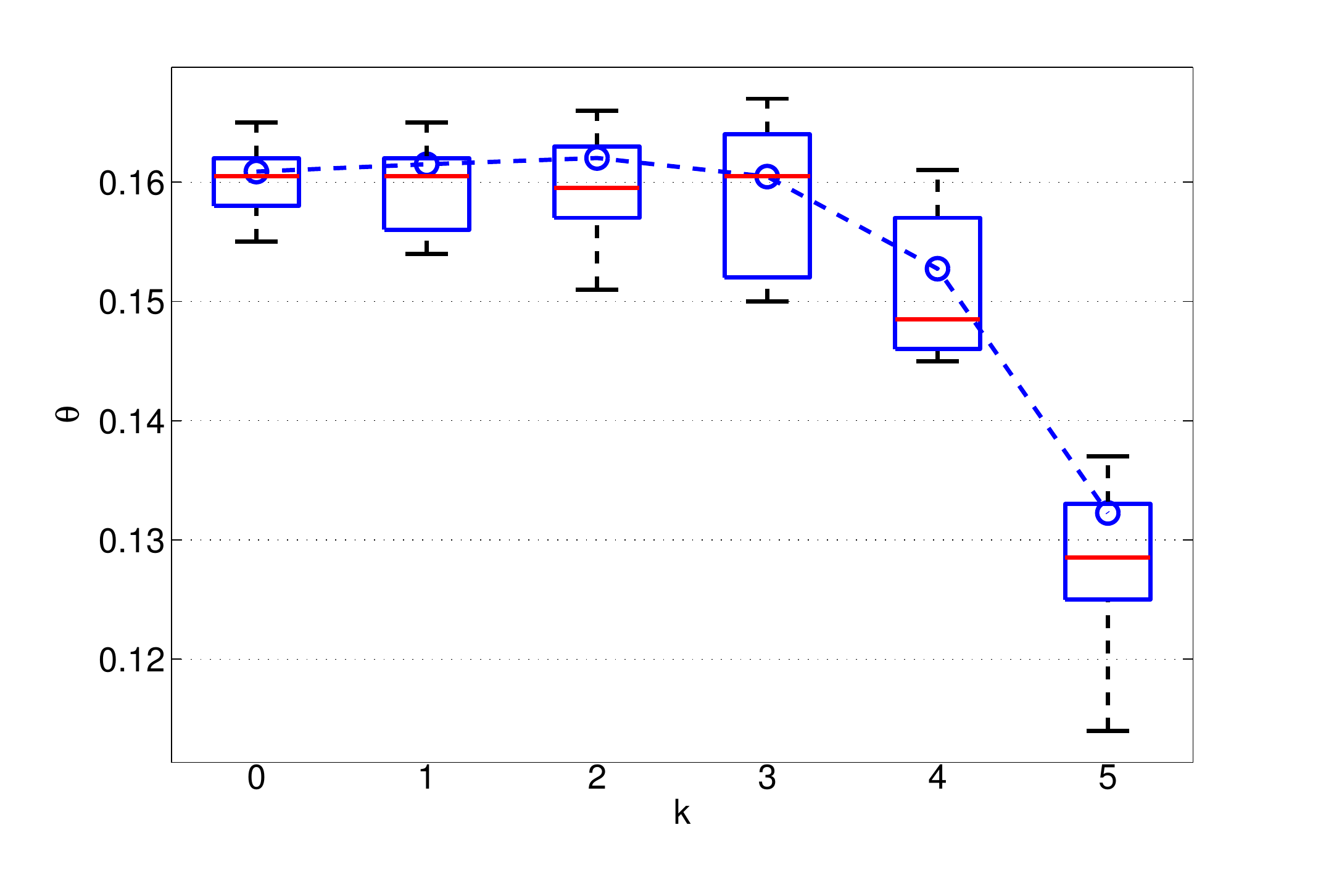}
     \caption{The evaluation of Eq.\ \eqref{eq:sr_ideal} along with the boxplot of the numerical results of 10 simulations for $N=1000$. The initial nodes' degree is uniformly distributed between $5-k$ and $5+k$, where $k$ is indicated in the abscissa. }\label{fig:ex_ideal_sr}
 \end{figure}


\subsection{Configuration model with Poisson distribution}\label{subsec:poisson}

When the degree distribution tends to a Poisson distribution, the differential equation can be greatly simplified using the
large amount of independence of the Erd\"os R�nyi graph.
We indeed showed in \cite{aap} that a state space collapse does happen in that case since the
number of explored nodes is itself Markov.
The differential equation reads:

\begin{gather}
    {d\over dt}\bar u_t = -\lambda \bar u_t\big(e^{-\nu\bar u_t}+\nu\bar u_t + (1-e^{-\nu\bar u_t})(1+\nu\bar u_t)\big)\nonumber\\
    \Rightarrow {d\over dt}\bar u_t = -\lambda\bar u_t\big(1+2\nu\bar u_t -e^{-\nu\bar u_t}\nu\bar u_t\big), \label{eq:eq_diff_er_ideal}
\end{gather}
and in this case the spatial reuse may be calculated as:
\begin{gather}\label{eq:theta_er_ideal}
    \theta = \lambda\int_0^\infty(1-e^{-\nu\bar u_t})\bar u_t dt, 
\end{gather}
where $\bar u_t$ is obtained from Eq.\ \eqref{eq:eq_diff_er_ideal}. 

Figure \ref{fig:theta_er_ideal} shows the corresponding spatial reuse for 100 contention periods (in the form of a boxplot) for different values of both $\nu$ and $N$, along with the estimation provided by Eq.\ \eqref{eq:theta_er_ideal}. Some interesting conclusions may be drawn from these graphs. 

 \begin{figure}
     \centering
     \includegraphics[scale=0.37]{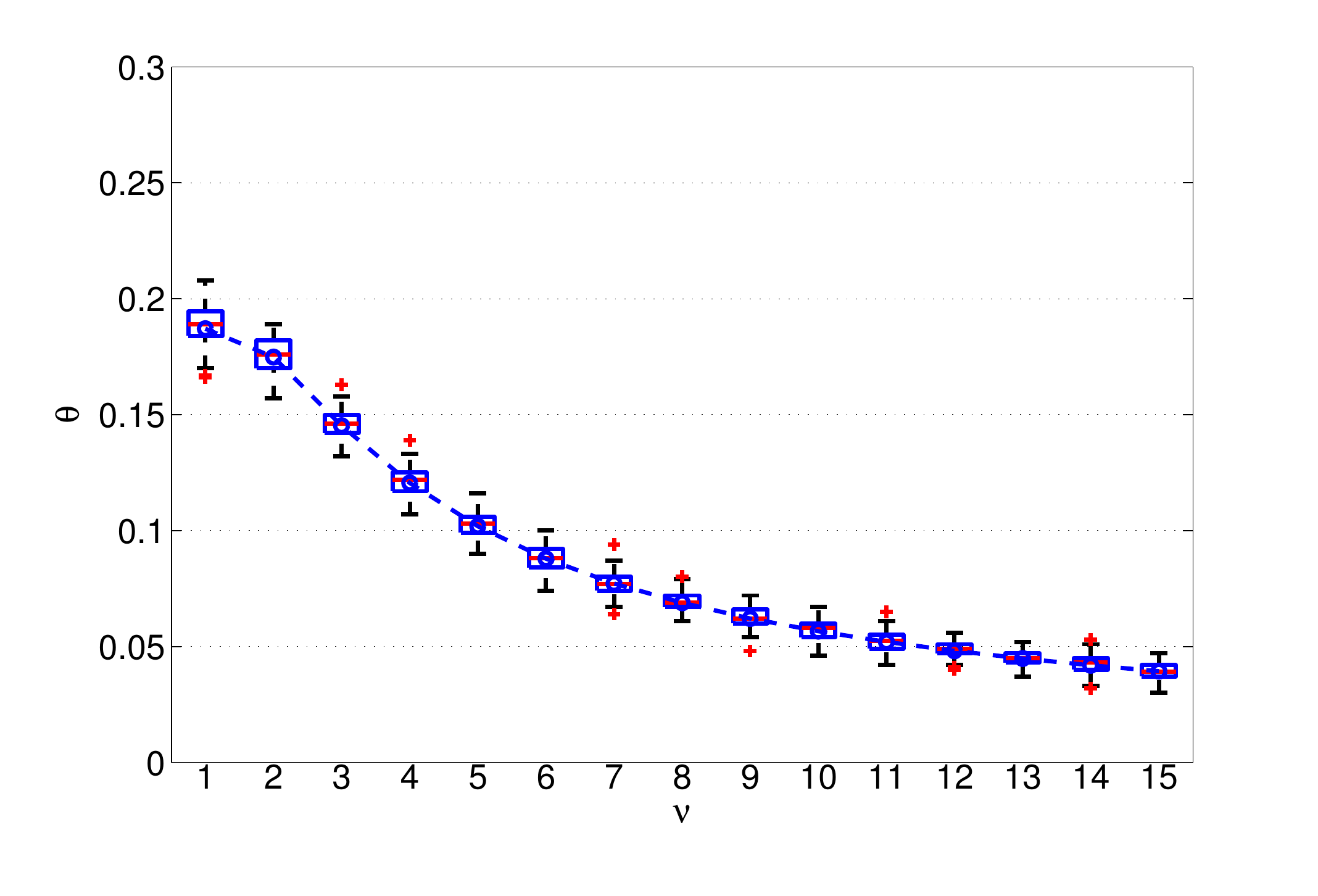}
     \includegraphics[scale=0.37]{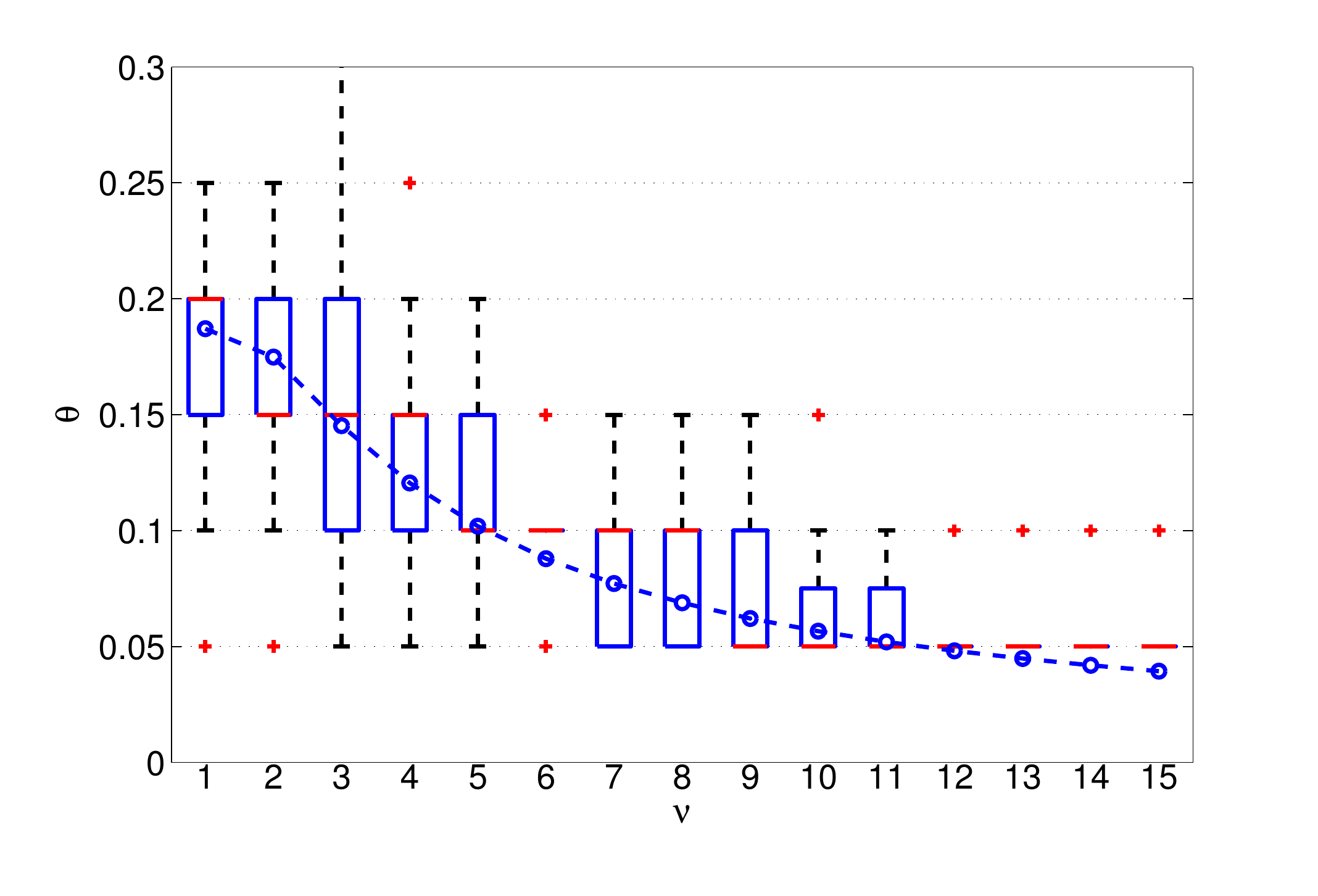}
     \caption{The evaluation of Eq.\ \eqref{eq:theta_er_ideal} along with the boxplot of the numerical results of 100 simulations for $N=1000$ (left) and $N=20$ (right). The initial nodes' degree is distributed as a Poisson with parameter $\nu$. }\label{fig:theta_er_ideal}
 \end{figure}

Firstly, when $N$ is large, the estimation is not only an excellent approximation to the mean spatial reuse, but the variance of $\Theta$ is very small. Secondly, as $N$ decreases, and although the variance increases significantly, Eq.\ \eqref{eq:theta_er_ideal} still provides an excellent approximation to the mean spatial reuse. This is further illustrated by Fig.\ \ref{fig:comp_theta_er_ideal}. Lastly, a comparison between Figs.\ \ref{fig:ex_ideal_sr} and \ref{fig:theta_er_ideal} shows again that the spatial reuse does not depend on the mean degree:
the present example obtains a spatial reuse of roughly 0.1 for a mean $\nu=5$ neighbors and a variance  $\nu=5$. This same variance would be obtained in the uniform case shown in Fig.\ \ref{fig:ex_ideal_sr} with $k$ between 3 and 4, which results in a spatial reuse of more than 0.15. 

 \begin{figure}
     \centering
     \includegraphics[scale=0.37]{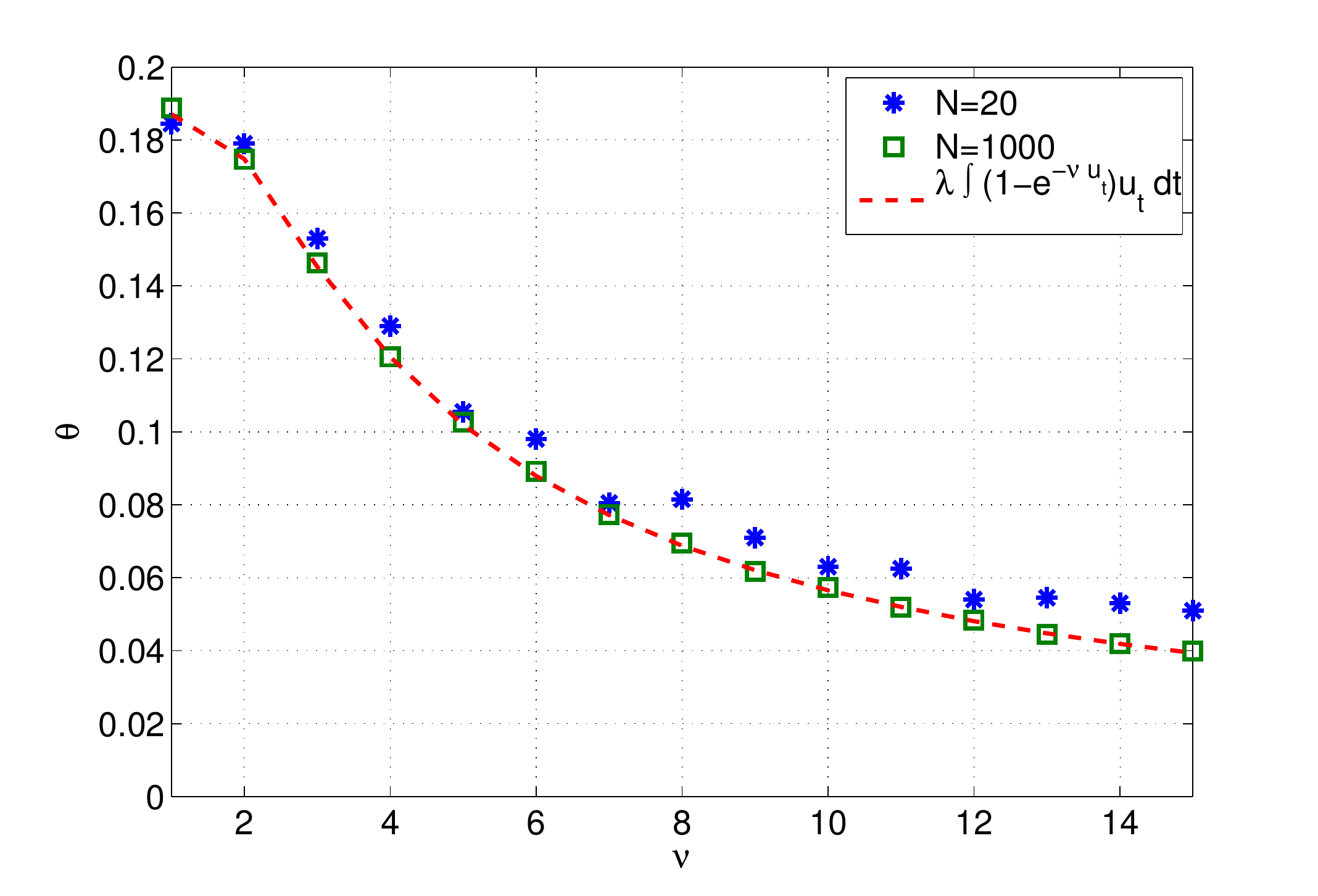}
     \caption{The mean of the numerical results of 100 simulations for $N=1000$ and $N=20$, along with the evaluation of Eq.\ \eqref{eq:theta_er_ideal}. The initial nodes' degree is distributed as a Poisson with parameter $\nu$. }\label{fig:comp_theta_er_ideal}
 \end{figure}


\subsection{Comparison with spatial reuses on fixed graphs}

We now discuss a scenario falling outside the scope of the initial assumptions 
to illustrate the efficiency of this method for a large class of models.
 Let us assume that the communication graph is not random, but fixed. 
 Thus, we may calculate the empirical distribution of neighbors of the initial graph (i.e.\ $\mu_0$) and calculate \eqref{eq:sr_ideal} by solving Eq.\ \eqref{eq:ecdif_ideal}. 
Though our method would consider a graph which is chosen randomly among all graphs that comply with the initial degree distribution, instead of a fixed graph,
the resulting spatial reuse still gives a reasonable approximation inmany cases. 

As a toy example, consider a lattice where every node has exactly $k$ neighbors. Figure \ref{fig:lattice_graphs} shows a portion of these graphs for $k=2,4$. The initial distribution in this case is $\mu_0(i)=\delta(i-k)$. 

 \begin{figure}
     \centering
     \includegraphics[scale=0.3]{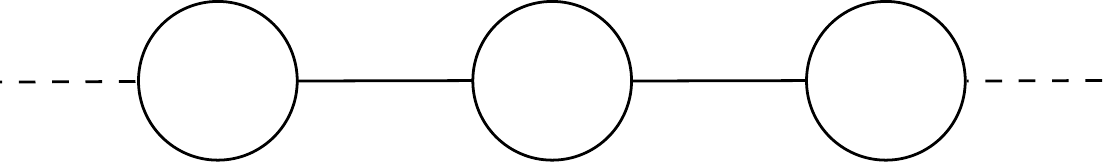}
     \includegraphics[scale=0.3]{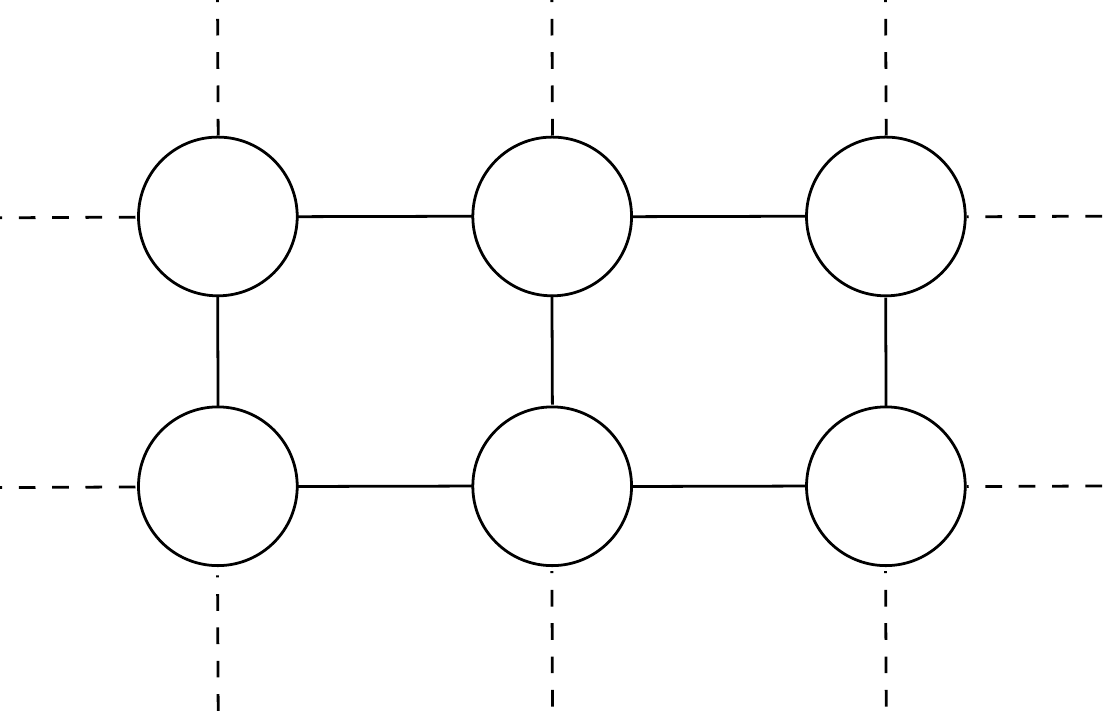}
     \caption{The graph lattices considered: a line and a grid.}\label{fig:lattice_graphs}
 \end{figure}

 For $k=4$, simulations indicate a spatial reuse of 0.17, whereas our estimation is 0.185. 
 Note that in the case of $k=2$, all random graphs generated by $\delta(i-k)$ are either as in Fig.\ \ref{fig:lattice_graphs} or several circles of interconnected nodes. Since as $N$ increases these circles include several nodes, our estimation is asymptotically exact.


\subsection{Comparison with the parking on a Poisson point process}\label{subsec:spa_poisson}


As previously underlined, our approach based on configurations models
essentially ignores correlations  between edges of the interference graph that are present when the graph  steams from a spatial model. 
We numerically show here that this effect is quantitatively very small as
soon as the interference graph has a sufficient amount of noise in the case of parking processes on a Poisson point process.
Consider then the ``classic'' model where nodes are randomly and uniformly located in a plane. A transmission with power $P$ of node $s$ is received at node $r$ with a \emph{mean} power $P\times L(d_{sr})$, where $d_{sr}$ is the distance between $s$ and $r$, and $L(\cdot)$ is a monotonous decreasing function (generally termed \emph{path loss}). This mean is taken over several time-slots. The receiving power during a given time-slot (which we will note $P(s,r)$) has random fluctuations around this mean, resulting in:
\begin{gather*}
    \frac{P(s,r)}{P} = L(d_{sr})\times X_{sr}, 
\end{gather*}
where $X_{sr}$ (generally termed \emph{fading}) is a random variable with mean value equal to one. 

Given a realization of the spatial process and the fading between every pair of nodes, the communication graph is constructed by including an edge between a pair of nodes $s$ and $r$ if $P(s,r)>P_{\min}$. $P_{\min}$ is the sensitivity of the receiver, indicating the minimum power that it requires to correctly decode a frame. We will assume that fading is symmetrical, so that the resulting channels are also symmetrical.

As an example, consider a path-loss function $L(d)=d^{-a}$ (with $a=2$) and log-normally distributed fading (whose logarithm is normally distributed with mean $0$ and variance $\sigma^2$). Nodes will be positioned in the plane as in a Poisson process with intensity 1, and $P/P_{\min}$ will be such that when $\sigma=0$ the mean number of neighbors of each node will be $\nu=2$ (i.e.\ $P/P_{\min}=(\pi/\nu)^{a/2}$). 

Figure \ref{fig:sr_hcmasruido} shows the results corresponding to this scenario for different values of $\sigma$. Please note that $\sigma=0$ corresponds to a variant of the so-called Mat\`ern hard core process~\cite{stoyan_matern}. Just like in the previous example, in this case considering only the empirical degrees' distribution results in a loss of information with a significant impact on the resulting spatial reuse. For instance, it is very likely that $\sTX$ and $\sRX$ have a neighbor in common. This in turn results in an underestimation of the spatial reuse, as shown in Fig.\ \ref{fig:sr_hcmasruido}. However, as $\sigma$ increases, this ``spatial correlation'' becomes weaker: the event of two nodes being neighbors is relatively less influenced by their distance. This results in increasingly more precise estimations of our method, which for a relatively small $\sigma=1$ already provides a very accurate estimation of the spatial reuse. 

 \begin{figure}
     \centering
     \includegraphics[scale=0.3]{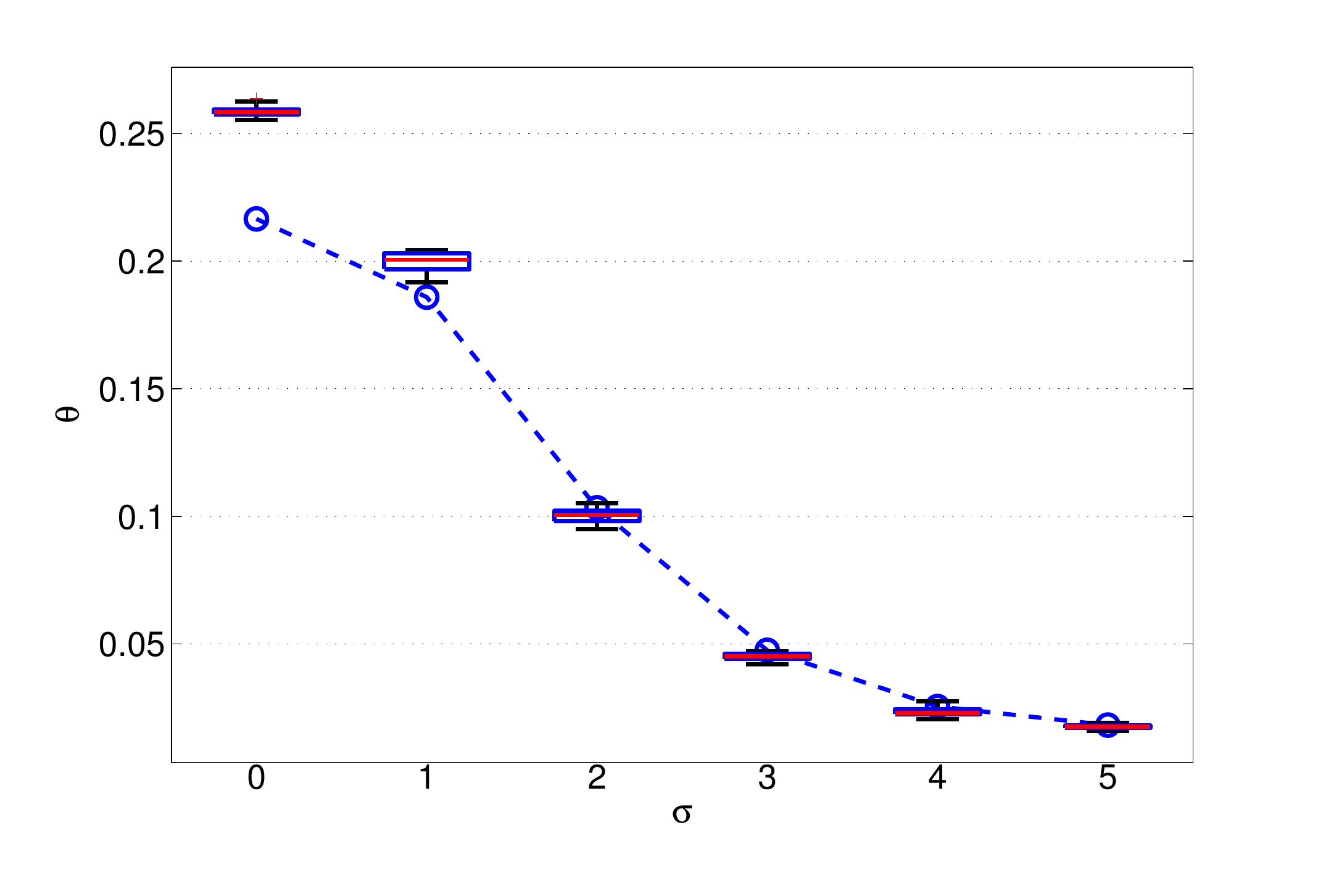}
     \caption{The evaluation of Eq.\ \eqref{eq:sr_ideal} along with the boxplot of the numerical results of 10 time-slot simulations for a Poisson process with log-normal fading and a path-loss of the form $L(r)=r^{-2}$. The value of $\sigma$ corresponds to the standard deviation of the underlying normal distribution. }\label{fig:sr_hcmasruido}
 \end{figure}

\section{Model extensions}\label{sec:extensions}

In the previous sections we have analyzed what we termed an ideal case, where the tagged node sends an RTS to every one of its neighbors until either one of them answers with a CTS or no more neighbors are left. Equivalently, this may be regarded as a situation where the tagged node's intended destination is always available. This could model for instance an opportunistic network where the actual destination of the RTS is any available node. If this is not the case, we are compelled to analyze what happens when the RTS/CTS handshake fails. 

In the following two subsection we briefly illustrate how our framework may be adapted with relative ease to study this scenario as well. More in particular, we will consider two possible situations when the tagged node is only interested in a single random neighbor. In the first one, if this node does not answer with a CTS, the neighbors of the tagged node (that overheard the RTS frame) will still be blocked for the rest of the contention period. In the second one, under this event, the neighbors of the tagged node, realizing that the data transmission did not start, ignore the RTS and are not blocked by this failed handshake. 

\subsection{RTS/CTS handshake failure}


Let us then consider the first scenario described above. For the sake of clarity, let us adapt the algorithm presented in Sec.\ \ref{subsec:naive} (which discussed how the \emph{unexplored}, \emph{active} and \emph{blocked} set evolved over time) and highlight the differences with the ideal case. Assume we are in a transition time $t_n$ when an unexplored node $s$ sends an RTS frame. The first step is still to update the active and unexplored sets as follow: $\A_{t_n^+}\leftarrow\A_{t_n}\cup\{s\}$ and $\U_{t_n^+}\leftarrow\U_{t_n}\setminus\{s\}$. 

From the set of neighbors of $s$ ($\mathcal{N}_s$), we choose a random node ($r$) if any. Then, the following two cases are possible: 
\begin{enumerate}
    \item If  $r \notin \U_{t_n^+}$ (or $\mathcal{N}_s=\emptyset$) is when this scenario differs from before. The set of active nodes remains unchanged and we have to update the set of blocked nodes so as to include only the neighbors of $s$: $\B_{t_n^+}\leftarrow \B_{t_n}\cup \mathcal{N}_{s}$. The set of unexplored nodes is updated accordingly; i.e.\ $\U_{t_n^+}\leftarrow \U_{t_n^+}\setminus \mathcal{N}_{s}$. We then proceed to the next node. 
    \item If $r\in \U_{t_n^+}$ we proceed exactly as before. That is to say, we include $r$ among the active nodes ($\A_{t_n^{+}}\leftarrow \A_{t_n^+}\cup \{r\}$), update the blocked set of nodes ($\B_{t_n^+}\leftarrow \B_{t_n}\cup (\mathcal{N}_{s}\setminus \{r\})\cup (\mathcal{N}_{r}\setminus \{s\})$) and the unexplored one ($\U_{t_n^+}\leftarrow \U_{t_n^+}\setminus \left(\mathcal{N}_{i}\cup\mathcal{N}_{r_i}\right)$). We then proceed to the next node. \label{paso:cts_rtsmata}
\end{enumerate}

Regarding the measured-value Markov chain approach, a measure state $\mu_t$ involving only the edges between unexplored nodes of the graph will not be enough. In particular, we now also need to track the number of blocked neighbours of the unexplored nodes to calculate the probability of successfully completing the RTS/CTS handshake, which in turn defines the amount of nodes that are removed from the measure and their degree. We propose then to define a bi-dimensional measure $\mu_t(i,j)$ representing the number of unexplored nodes with $i$ half-edges toward the unexplored set ($\sU\to\sU$ type) and $j$ half-edges toward the blocked set ($\sU\to\sB$ type) at time $t$. 

The initial measure will be then defined as:
\begin{gather*}
\mu_0(i,j)=
\begin{cases}
    h(i)& \mbox{if} \quad  j=0,\\
       0& \mbox{if} \quad  j>0.
\end{cases}
\end{gather*}
That is to say, at $t=0$ all nodes points towards unexplored nodes, and the measure is given by the nodes' degree histogram. Please note that $\sum_j\mu_t(i,j)$ yields the measure we had in the ideal case. 

It is relatively straightforward to verify that this measure will be enough to characterize the spatial reuse in this case. For instance, and following the notation we used in Sec.\ \ref{subsec:conf_algo}, since TX is randomly chosen from all the unexplored nodes, its degree towards $\U_t$ and $\B_t$ has joint probability distribution given by:
\begin{gather*}
    P(K^{\TX}_{\sU}=i,K^{\TX}_{\sB}=j) = \alpha_t(i,j) = \frac{\mu_t(i,j)}{\sum_{k,l}\mu_t(k,l)}, 
\end{gather*}
where we added the subscripts $\sU$ and $\sB$ to indicate the type of half-edge. Moreover, all the unexplored neighbors of an unexplored node (and in particular $RX$) have the following degree distribution: 
\begin{gather*}
    P(K^{\RX}_{\sU}=i,K^{\RX}_{\sB}=j) = \beta_t(i,j) = \frac{i\mu_t(i,j)}{\sum_{k,l}k\mu_t(k,l)}. 
\end{gather*}

As in the previous case, each step of the algorithm induces a modification on the measure $\mu_t$. If again we scale it by considering $\mu_t^N=\mu_t/N$ and take the limit  when $N$ goes to infinity we obtain a fluid limit  for the evolution of the measure $\mu_t$. This limit is a deterministic measure $\bar\mu_t$ which (as before) will be the solution of an infinite dimensional differential equation system. Following the discussion we presented for the ideal case, the right-hand of the equation for the evolution of $\mu_t(i,j)$ should be the mean number of nodes with $i$ half-edges of type $\sU\to\sU$  and $j$ half-edges of type $\sU\to\sB$ that are removed at time $t$, times the normalized total transition rate.
 
In this case the differential equation thus results: 
\begin{multline}\label{eq:ecdif_mataRTS_exp}
    {d \over dt } \bar\mu_t(i,j) = -\lambda \sum_{k,l} \bar\mu_t(k,l) \Bigg[ \bar\alpha_t(i,j) + P_t(CTS^c) \mathbb{E}\{Y^{\TX}(i,j) + X^{\TX}(i,j) - X^{\TX}(i+1,j-1)\,|\,CTS^c\}  \)\Bigg. \\ 
    \Bigg. + P_t(CTS)  \( \bar\beta(i,j) + \mathbb{E}\{ Y^{\TX,\RX}(i,j)+X^{\TX,\RX} (i,j) - X^{\TX,\RX} (i+1,j-1)\,|\,CTS \} \)\Bigg], 
\end{multline}
where by $CTS$ we refer here to the event described in step \ref{paso:cts_rtsmata} above (i.e. the selected receiver is unexplored), and the superscript $c$ refers to its complement. Thus we have that: 
\begin{gather*}
    P_t(CTS) =\sum_{k,l}P_t(CTS\,|\,K^{\TX}_{\sU}=k,K^{\TX}_{\sB}=l)\bar\alpha_t(k,l)=\sum_{k>0,l}\frac{k}{k+l}\bar\alpha_t(k,l). 
\end{gather*}
Please note that the main difference between Eqs.\ \eqref{eq:ecdif_mataRTS_exp} and \eqref{eq:ecdif_ideal_exp} lies in the definition of the event $CTS$ and in the fact that $\sTX$ blocks its neighbors even when this event does not occur (the second term between brackets). 

The expected value of the rest of the random variables may be obtained analogously to how we proceeded before. For instance, the degree of the neighbors of $\sTX$ ($Y^{\TX}(i,j)$) follows a hypergeometric distribution with parameters $K^{\TX}$ and $\bar\beta_t(i,j)$. Then, the mean value of $Y^{\TX}$ results:
\begin{gather*}
    \E \{Y^{\TX}(i,j)\,|\,CTS^c\} = \E\{K^{\TX}_{\sU}\,|\,CTS^c\} \bar\beta_t(i,j). 
\end{gather*}
Moreover, the number of nodes whose degree should be updated (the unexplored ``neighbors of the neighbors'') are $Z^{\TX}=\sum_{k,l}(k-1)Y^{\TX}(k,l)$, and we thus have that: 
\begin{gather*}
    \E\{X^{\TX}(i,j)\,|\,CTS^c\}=\beta_t(i,j)\sum_{k,l}(k-1)\E\{Y^{\TX}(k,l)\,|\,CTS^c\}. 
\end{gather*}

All in all, Eq.\ \eqref{eq:ecdif_mataRTS_exp} may be written in terms of $\bar\mu_t$, $\bar\alpha_t$ and $\bar\beta_t$ as follows:
\begin{multline}
    {d \over dt } \bar\mu_t(i,j) = -\lambda \sum_{k,l} \bar\mu_t(k,l) \left[ \bar\alpha_t(i,j) + \bar\beta_t(i,j)\sum_k k\bar\alpha_t(k)  + \left(\sum_{k>0,l} \frac{k}{k+l}\bar\alpha_t(k,l)\right) \bar\beta_t(i,j) \( \sum_{k,l} (k-1)\bar\beta_t(k,l)\) \right. \label{eq:ecdif_cts_mata} \\ 
        + \left. (\bar\beta_t(i,j)- \bar\beta_t(i+1,j-1) ) \sum_{k,l} (k-1)\bar\beta_t(k,l) \left( \sum_{k,l} k\bar\alpha_t(k,l) + \left(\sum_{k>0,l} \frac{k}{k+l}\bar\alpha_t(k,l) \right) \(\sum_{k,l} (k-2)\bar\beta_t(k,l)\) \right)
\right].
\end{multline}

Once we solved the previous differential equations we can estimate the spatial reuse as we did in Sec.\ \ref{subsec:spatial_reuse}. We obtain then the following relation:
\begin{gather}\label{eq:sr_bloqrts}
    \E\{\Theta\} = \theta = \lambda\int_0^\infty \bar u_t P_t(CTS) dt = \lambda\int_0^\infty \sum_{k,l}\bar\mu_t(k,l)\sum_{k>0,l}\frac{k}{k+l}\bar\alpha_t(k,l)dt, 
\end{gather}
where $\bar\mu_t$ is the solution of \eqref{eq:ecdif_cts_mata}. 

All the discussion we presented in the previous sections also applies to this case. For instance, and as example of both the precision and the limitations of our approach, Fig.\ \ref{fig:sr_hcmasruido_bloqrts} shows the results for this case in the same scenario as in Sec.\ \ref{subsec:spa_poisson} (Poisson hard core process). Again, the information lost by considering only the initial nodes' degree distribution may have a significant impact (small values of $\sigma$ in this case). If this is not the case, our approach obtains very precise results. 

 \begin{figure}
     \centering
     \includegraphics[scale=0.3]{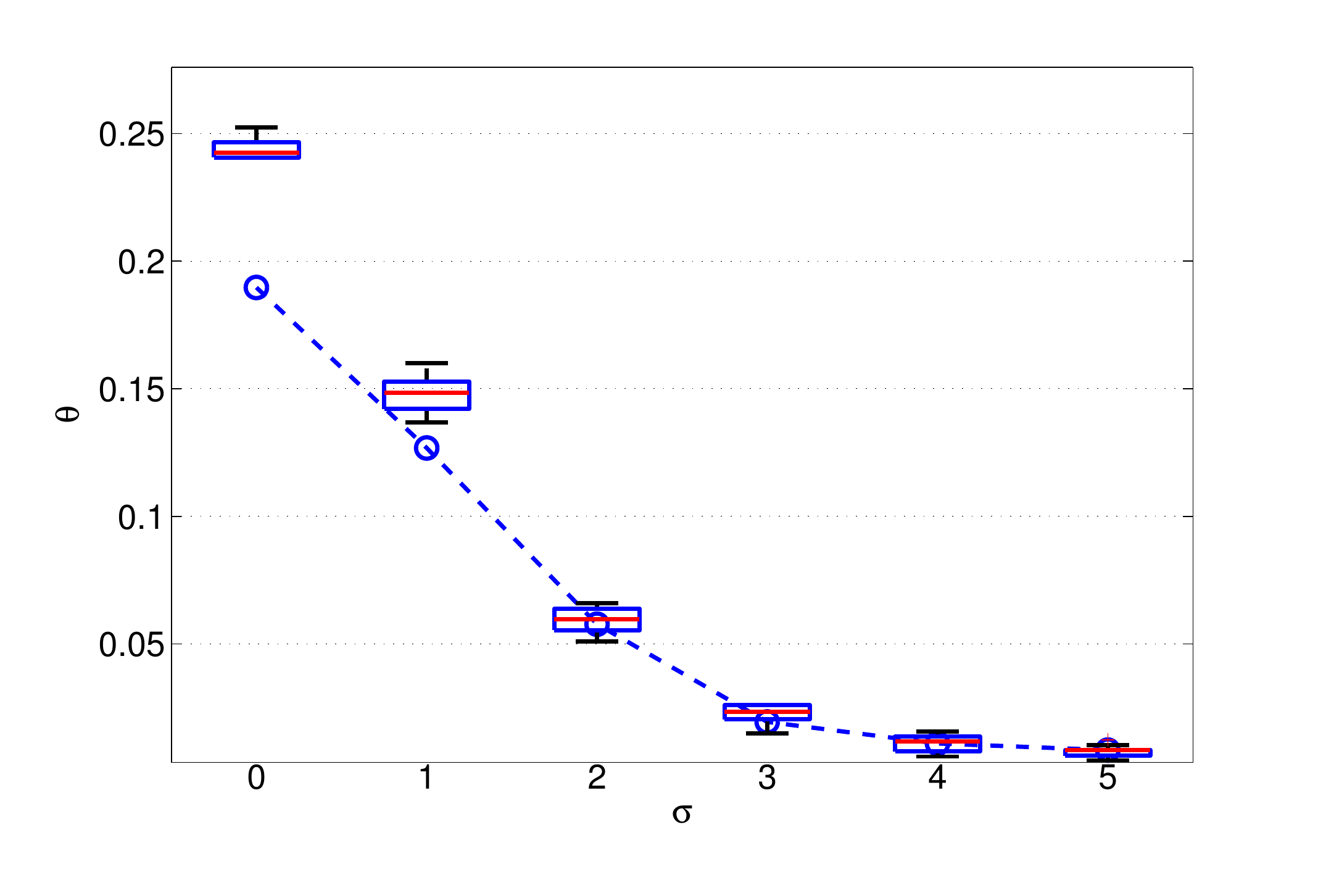}
     \caption{The evaluation of Eq.\ \eqref{eq:sr_bloqrts} along with the boxplot of the numerical results of 10 time-slot simulations for a Poisson process with log-normal fading and a path-loss of the form $L(r)=r^{-2}$. The value of $\sigma$ corresponds to the standard deviation of the underlying normal distribution. }\label{fig:sr_hcmasruido_bloqrts}
 \end{figure}

\subsection{RTS/CTS handshake failure with timeout}

Let us study the second scenario under an unanswered RTS frame. In the previous subsection, the neighbors of TX were blocked by the RTS and did not compete further for the channel, even if no data transmissions ensued. However, a time-out is generally implemented in this kind of access, where the RTS alone blocks the nodes during a certain time. If no further transmissions are sensed afterwards, the RTS is ignored and the node starts competing for the medium again. Let us then consider an idealization of this mechanism, where this realization is instantaneous. As we will discuss later, the presented extension is also capable of modeling a situation where the tagged node has no packets to send and acts only as a receiver. 

Please note that in this case, nodes whose RTS frame went unanswered will in turn be able to answer with a CTS if a RTS frame is addressed to them. Thus, they belong neither to the blocked nor active set of nodes. In this section we define a new class of nodes: \emph{sans}-CTS. We will say the node belongs to class $\sS$ (and the corresponding set $\S_t$) if it is available only as a receiver (it has tried to communicate without success). 

To highlight the differences with the previous scenario, we discuss here the different possibilities that arises when an unexplored node $s$ tries to communicate with a randomly selected neighbour at time $t_n$. The first step this time is to update only the set of unexplored nodes $\U_{t_n^+}\leftarrow\U_{t_n}\setminus\{s\}$. Once we have chosen a random neighbor of $s$ ($r\in\mathcal{N}_s$), the following two cases are possible: 
\begin{enumerate}
    \item If  $r \notin \U_{t_n^+} \cup \S_{t_n}$ was chosen (or if $\mathcal{N}_s=\emptyset$), the set of active nodes remains unchanged. Moreover, the node $s$ is still available as a receiver and its neighbors are not blocked. That is to say, $\S_{t_n^{+}}\leftarrow \S_{t_n}\cup \{s\}$ and we proceed to the next node. 
    \item If  $r\in \U_{t_n^+} \cup \S_{t_n}$  was chosen (i.e.\ the neighbour is unexplored or available), we re-update the sets as before, and including $s$ in the active set: $\A_{t_n^{+}}\leftarrow \A_{t_n}\cup \{s,r\}$,  $\B_{t_n^+}\leftarrow \B_{t_n}\cup (\mathcal{N}_{s}\setminus \{r\})\cup (\mathcal{N}_{r}\setminus \{s\})$ and $\U_{t_n^+}\leftarrow \U_{t_n^+}\setminus \left(\mathcal{N}_{s}\cup\mathcal{N}_{r}\right)$. Moreover, the involved nodes that belonged to $\S_{t_n}$ should be removed from it: $\S_{t_n^+}\leftarrow \S_{t_n}\setminus \left(\mathcal{N}_s\cup\mathcal{N}_r\right)$. We then proceed to the next node.
\end{enumerate}

The measure-valued Markov chain approach is more involved in this scenario. To begin with, we need to define a three-dimensional measure $\mu_t(i,j,k)$ to keep track of the degree of a given unexplored node towards the unexplored, blocked and \emph{sans}-CTS nodes. Moreover, we also need the information about the degree of the \emph{sans}-CTS nodes towards the unexplored and \emph{sans}-CTS sets. This is necessary since, for instance, once a node is chosen as a receiver, it will block its unexplored neighbours. This receiver may belong either to $\U_t$ or $\S_t$. 

In brief, we need two measures: $\mu_t(i,j,k)$ and $\nu_t(i,j)$. The measure $\mu_t(i,j,k)$ represents at time $t$ the number of unexplored nodes with $i$ half-edges toward the unexplored set ($\sU\to\sU$ type), $j$ half-edges toward the blocked set ($\sU\to\sB$ type) and $k$ half-edges toward the \emph{sans}-CTS set ($\sU\to\sS$ type). Analogously, $\nu_t(i,k)$ represents at time $t$ the number of \emph{sans}-CTS nodes with $i$ half-edges toward the unexplored set ($\sS\to\sU$ type) and $k$ half-edges towards the \emph{sans}-CTS set ($\sS\to\sB$ type). At time $t=0$ these measures are:
\begin{gather*}
\mu_0(i,j,k)=
\begin{cases}
    h(i)& \mbox{if} \quad  j=k=0,\\
       0& \mbox{if} \quad  j,k>0,
\end{cases}
\quad \text{ and } \quad
\nu_0(i,k)=0. 
\end{gather*}
That is to say, initially all nodes are unexplored and point towards unexplored nodes. Note that we could analyze a situation where some nodes act only as receivers, by assigning them to the initial measure $\nu_0$ (and reflecting this on $\mu_0$). 

We then analyze a large graph limit when the number of nodes goes to infinity as we did for the previous cases. The procedure is essentially the same: to write down the evolution of $\mu_t(i,j,k)$ (respectively $\nu_t(i,j)$) we identify at each step of the algorithm the mean number of nodes of each type and degree that should be removed from the measure. The random variables we defined before (notably $Y$ and $X$) have to be adapted to this case, but their distribution is essentially the same. For the sake of clarity, and since our objective here is to illustrate possible extensions of our framework, we do not write the differential equations system. 

Once we have obtained a (numerical) solution of the corresponding equation we can estimate as before the spatial reuse as:
\begin{gather}\label{eq:sr_nobloqrts}
        \E\{\Theta\} = \theta = \lambda\int_0^\infty \bar u_t P_t(CTS) dt.
\end{gather}
where $\bar u_t= \sum\limits_{i,j,k}\bar\mu_t(i,j,k)$ and $P_t(CTS) = \sum\limits_{i+k>0,j} \frac{i+k}{i+j+k}  \bar\alpha_t(i,j,k)$. 

Again, and for illustrative purposes, we consider the example of the Poisson hard core process. The results, presented in Fig.\ \ref{fig:sr_hcmasruido_nobloqrts}, again show that when the initial nodes' degree distribution is enough to describe the statistics of the resulting communication graph (corresponding to the higher values of $\sigma$), then our method yields very precise results.  

 \begin{figure}
     \centering
     \includegraphics[scale=0.3]{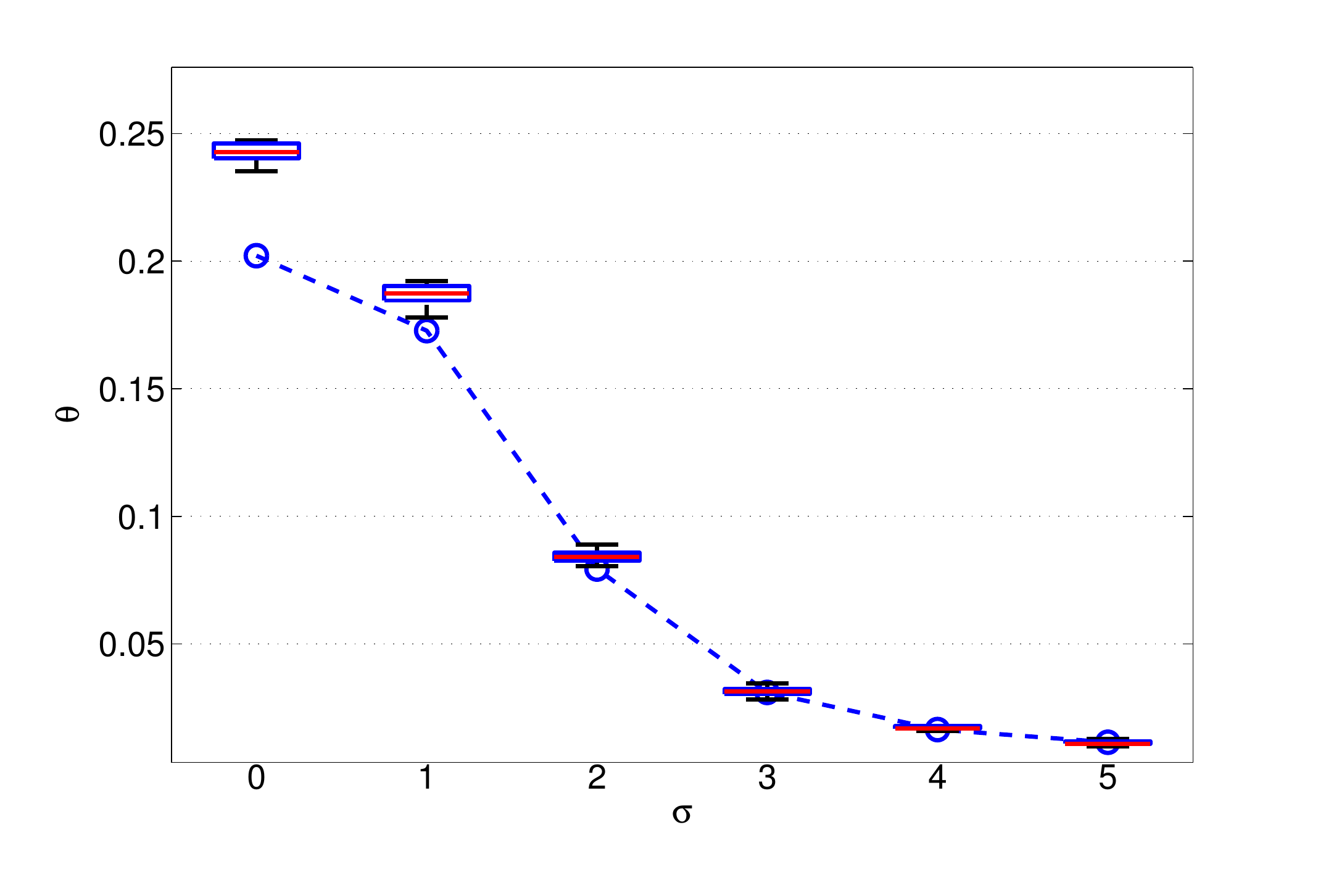}
     \caption{The evaluation of Eq.\ \eqref{eq:sr_nobloqrts} along with the boxplot of the numerical results of 10 time-slot simulations for a Poisson process with log-normal fading and a path-loss of the form $L(r)=r^{-2}$. The value of $\sigma$ corresponds to the standard deviation of the underlying normal distribution. }\label{fig:sr_hcmasruido_nobloqrts}
 \end{figure}

\section{Conclusions and future works}

We showed how stochastic dynamics on configuration models
might provide powerful performance tools for estimating the spatial
reuse of wireless networks. 
Error bounds and stochastic comparison results for spatial models would be of great practical
importance.

\textbf{References}

\bibliographystyle{elsarticle-num}
\bibliography{bibliography}
\end{document}